\documentclass[letterpaper, 10 pt, conference]{ieeeconf}  % Comment this line out if you need a4paper

\IEEEoverridecommandlockouts                              % This command is only needed if 
\let\labelindent\relax
\usepackage{enumitem}
\usepackage{times}
\usepackage{amsmath}
\usepackage{amssymb}
\usepackage{graphicx, url, amsthm}
\usepackage{balance}

\newtheorem*{example}{Running example}
\theoremstyle{definition}
\newtheorem{proposition}{Proposition}
\newtheorem{definition}{Definition}
\newtheorem{corollary}{Corollary}
\newtheorem{lemma}{Lemma}
\newtheorem{theorem}{Theorem} 
\newtheorem{remark}{Remark}
\newtheorem{assumption}{Assumption}

\usepackage{caption}
\usepackage{booktabs}
\usepackage{multirow}

\usepackage{subcaption}
\usepackage[ruled,vlined,linesnumbered]{algorithm2e}
\DontPrintSemicolon
\SetAlgoSkip{}

\makeatletter
    \renewcommand{\@algocf@capt@plain}{above}
    \makeatother

\newcommand{\state}{{x}}
\newcommand{\control}{{u}}
\newcommand{\stagecost}{{c}}
\newcommand{\sumcost}{{J}}
\newcommand{\lqsumcost}{{J}}
\newcommand{\valuefunc}{{V}}
\newcommand{\dyn}{{f}}
\newcommand{\totalnumberofplayer}{{N}}
\newcommand{\stage}{{t}}

\newcommand{\totalstatedim}{{n}}
\newcommand{\totalcontroldim}{{m}}

\newcommand{\initialstatedistribution}{{\mathcal{D}}}
\newcommand{\linearpolicy}{{K}}
\newcommand{\policy}{{\pi}}
\newcommand{\policyofplayeri}{{\pi^i}}

\newcommand{\optimalpolicyofplayeri}{{\pi^{i*}}}

\newcommand{\statecostmatrix}{{Q}}
\newcommand{\controlcostmatrix}{{R}}
\newcommand{\statematrix}{{A}}
\newcommand{\controlmatrix}{{B}}

\newcommand{\nominalstate}{{\bar{x}}}
\newcommand{\nominalcontrol}{{\bar{u}}}
\newcommand{\NNpolicy}{{\pi_\theta}}

\newcommand{\stepsize}{{\eta}}

\newcommand{\statecostlinearterm}{{q}}
\newcommand{\controlcostlinearterm}{{r}}

\newcommand{\expectationoverinit}{\mathbb{E}_{x_0\sim\initialstatedistribution}}

\newcommand{\sumcostregularized}{\hat{J}_{\rho}}

\usepackage{color}

\title{\LARGE \bf
Multi-Agent Guided Policy Search for Non-Cooperative Dynamic Games
}

\author{Jingqi Li$^*$, Gechen Qu$^{*}$, Jason J. Choi, Somayeh Sojoudi, Claire Tomlin
\thanks{*Equal contribution. }% <-this % stops a space
\thanks{Jingqi Li is with the University of Texas at Austin. Gechen Qu, Jason Choi, Somayeh Sojoudi, and Claire Tomlin are with the University of California, Berkeley. Corresponding author: Jingqi Li {\tt\small jingqi.li@austin.utexas.edu}}
}

\begin{document}

\maketitle
\thispagestyle{empty}
\pagestyle{empty}

%%%%%%%%%%%%%%%%%%%%%%%%%%%%%%%%%%%%%%%%%%%%%%%%%%%%%%%%%%%%%%%%%%%%%%%%%%%%%%%%
\begin{abstract}
%Non-cooperative dynamic games 
Multi-agent reinforcement learning (MARL) optimizes strategic interactions in non-cooperative dynamic games, where agents have misaligned objectives. However, data-driven methods such as multi-agent policy gradients (MA-PG) often suffer from instability and limit-cycle behaviors. Prior stabilization techniques typically rely on entropy-based exploration, which slows learning and increases variance. We propose a model-based approach that incorporates approximate priors into the reward function as regularization. In linear quadratic (LQ) games, we prove that such priors stabilize policy gradients and guarantee local exponential convergence to an approximate Nash equilibrium. We then extend this idea to infinite-horizon nonlinear games by introducing Multi-agent Guided Policy Search (MA-GPS), which constructs short-horizon local LQ approximations from trajectories of current policies to guide training. Experiments on nonlinear vehicle platooning and a six-player strategic basketball formation show that MA-GPS achieves faster convergence and more stable learning than existing MARL methods.
% Non-cooperative dynamic games model strategic interactions among agents with potentially conflicting objectives. Multi-agent policy gradient (MA-PG) methods for these games often suffer from unstable limit-cycle behaviors. Unlike prior stabilization methods relying on random exploration, we propose using an approximate model-based prior as reward regularization. We theoretically show that this prior stabilizes policy gradients and enables local exponential convergence toward Nash equilibria in linear quadratic (LQ) games. Building upon this, we extend our approach to infinite-horizon nonlinear games by introducing Multi-agent Guided Policy Search (MA-GPS), which leverages short-horizon local LQ approximations to guide training. Experiments involving three-vehicle platooning, and a six-agent ``basketball game,'' demonstrate that MA-GPS has a faster and stabler learning than existing MARL methods.
\end{abstract}

%%%%%%%%%%%%%%%%%%%%%%%%%%%%%%%%%%%%%%%%%%%%%%%%%%%%%%%%%%%%%%%%%%%%%%%%%%%%%%%%
\section{Introduction}

Non-cooperative dynamic games \cite{bacsar1998dynamic} provide a principled framework for modeling strategic interactions among self-interested agents, with applications in robotic coordination \cite{maravall2013coordination}, autonomous driving \cite{zhou2022multi}, and human–robot interaction \cite{yu2021optimizing}. A central solution concept is the Nash equilibrium \cite{bacsar1998dynamic}, which captures unilateral optimal decision-making among agents. Multi-agent reinforcement learning (MARL) has emerged as a promising data-driven tool for computing Nash equilibria in high-dimensional games. However, learning Nash equilibria remains challenging due to non-stationarity from simultaneous policy updates and conflicting objectives of agents. As a result, multi-agent policy gradients (MA-PG) often exhibit unstable dynamics or converge to limit cycles \cite{mazumdar2020policy}, even under smooth objectives and dynamics.

Classical approaches instead rely on model-based optimization \cite{bacsar1998dynamic,di2019ddp}. For example, \cite{fridovich2020efficient} and \cite{laine2023computation} compute local Nash equilibrium using iterative linear quadratic (LQ) game approximations and Newton-type algorithms. While these methods usually converge stably, they scale poorly as the number of agents increases and the dynamics become nonlinear with complex objectives. In such cases, they become computationally expensive and may converge to inferior policies \cite{kossioris2008feedback}.

These limitations highlight the complementary strengths of MARL and model-based optimization, suggesting that a hybrid method could provide both scalability and stability. In this paper, we move beyond prior efforts that integrate such approaches for cooperative multi-agent control \cite{jiang2021mrgps,li2025multi} and ask: \emph{how can we integrate these two approaches to improve policy learning for non-cooperative dynamic games?}

Our key idea is to augment MA-PG methods with a model-based prior, defined as an approximate Nash equilibrium solution derived from local models (e.g., LQ approximations). This prior provides a shared reference that stabilizes policy updates across agents, avoids the limit cycles of pure MA-PG methods caused by non-unique Nash equilibria, simultaneous learning, and conflicting objectives, and ultimately guides agents toward a consistent Nash equilibrium.

We propose a framework that incorporates this prior into the reward function to regularize MA-PG updates. This constitutes a new mechanism for stabilizing MARL, outperforming existing MARL stabilization approaches that primarily rely on entropy-driven exploration~\cite{hambly2023policy, aggarwal2024policy}. Beyond faster convergence, our method yields neural network policies that can be executed in real time for complex multi-agent dynamic games, where classical model-based optimization methods are often computationally demanding for timely decision-making. Our contributions are:
% We propose a framework that incorporates this prior into the reward function to regularize MA-PG updates. In doing so, it offers a new mechanism for stabilizing MARL that outperforms existing MARL stabilization techniques based on entropy-driven exploration~\cite{hambly2023policy, aggarwal2024policy}. Our method not only achieves faster convergence but also yields neural network policies that can be executed in real time for complex multi-agent dynamic games, where classical model-based optimization methods are computationally demanding for timely decision-making. Our contributions are:

\begin{enumerate}
\item For LQ games, characterized by linear dynamics and quadratic objectives, we theoretically prove that any stabilizing feedback policy, not necessarily to be a Nash equilibrium policy, can serve as effective guidance and establish formal guarantees on the stability and local convergence of guided policy gradients toward an approximate Nash equilibrium.
\item We extend the framework to nonlinear games by leveraging the fact that Nash equilibria of local LQ game approximations closely resemble those of the original nonlinear games. To construct these approximations efficiently, we use trajectories generated by the current neural network policies and treat the Nash equilibrium of the corresponding local LQ game approximation as guidance, eliminating the need for precomputed guidance policies used in classical guided policy search.
\item We evaluate our method on LQ games, nonlinear vehicle platooning, and a six-player strategic basketball formation, demonstrating faster and more stable convergence with higher-quality policies than state-of-the-art MARL methods in dynamic games with continuous state and action spaces.
\end{enumerate}

\section{Related works}

\noindent \textbf{Model-based solutions to non-cooperative games.} Computing Nash equilibria of non-cooperative games is computationally challenging \cite{daskalakis2009complexity}; however, under smoothness assumptions on dynamics and costs, they can be solved in polynomial time \cite{bacsar1998dynamic}. For instance, LQ games enable efficient computation, with feedback Nash equilibria obtained via coupled Riccati equations or KKT conditions \cite{laine2023computation}. Extending these solutions to general nonlinear games remains difficult. Recent advances include iLQGames, which applies local LQ models and dynamic programming in an iterative scheme \cite{fridovich2020efficient}, and differential dynamic programming \cite{di2019ddp}. Other approaches compute equilibria using Newton’s method for KKT systems \cite{facchinei2009generalized,laine2023computation} or approximate dynamic programming \cite{epelman2011sampled,perolat2015approximate}, but high-dimensional nonlinear games remain challenging.

% \vspace{0.5em}

\noindent \textbf{Data-driven MARL methods.} 
MARL has been extensively studied in cooperative settings \cite{tan1993multi,claus1998dynamics,tampuu2017multiagent,oroojlooy2023review,li2025multi}, with notable successes in actor–critic methods \cite{konda1999actor,ackermann2019reducing,de2020independent,yu2022surprising}, value decomposition \cite{sunehag2018value,rashid2020monotonic}, and graph neural network approaches \cite{nayak2023scalable}. In contrast, learning game-theoretic equilibria \cite{lowe2017multi} (e.g., Nash equilibria) with deep RL remains challenging due to coupled dynamics and conflicting objectives. Policy gradient methods and their variants converge reliably in two-agent zero-sum games \cite{zhang2019policy,daskalakis2020independent,sokota2022unified} and potential games \cite{hosseinirad2023policy}, where agents’ costs are symmetric. In general non-cooperative games, however, they may converge to limit cycles rather than a Nash equilibrium \cite{mazumdar2020policy}. Furthermore, directly applying guided policy search \cite{levine2013guided}, which improves training in single-agent RL, to multi-agent Nash equilibrium settings is nontrivial due to misaligned objectives, except in fully cooperative cases \cite{jiang2021mrgps,li2025multi}.

% \vspace{0.5em}

\noindent \textbf{Training stability and variance reduction in deep RL.} 
Introducing random noise into the dynamics can stabilize MA-PG methods in non-cooperative games \cite{hambly2023policy, aggarwal2024policy}, suggesting that stochastic policies may improve MARL performance. However, this stabilization comes with increased gradient variance. Variance reduction techniques are well-established in single-agent RL \cite{sutton2018reinforcement}, focusing on managing the bias-variance trade-off \cite{thomas2014bias, schulman2015high, gu2022q}. In multi-agent contexts, existing variance reduction methods include optimal baselines \cite{kuba2021settling} and spectral normalization of critics \cite{mehta2023effects}, though these approaches do not incorporate model-based priors.

Model-based priors have been explored to reduce variance in RL \cite{cheng2019control}. Unlike \cite{cheng2019control}, which regularizes the policy output directly, our method incorporates regularization into the reward function. Moreover, prior works such as \cite{hambly2023policy, aggarwal2024policy,lidard2024blending} stabilize MA-PG in finite-horizon games using entropy-driven exploration. In contrast, we show that suitable policy guidance ensures stability in infinite-horizon settings, which could be more practical for real-time decision-making since the optimal policy is time-invariant and avoids computing complex time-varying Nash equilibrium policies.

\section{Background: non-cooperative dynamic games}
\subsection{Problem settings}
We consider an $\totalnumberofplayer$-agent discrete-time non-cooperative dynamic game. At each stage $\stage$, we associate with each agent $i\in[\totalnumberofplayer]:=\{1,\dots,\totalnumberofplayer\}$ a state vector $\state_\stage^i\in\mathbb{R}^{\totalstatedim_i}$ and a control input vector $\control_\stage^i\in\mathbb{R}^{\totalcontroldim_i}$. Let $\totalstatedim:=\sum_{i=1}^\totalnumberofplayer \totalstatedim_i$ be the dimension of the joint state, and $\totalcontroldim:=\sum_{i=1}^\totalnumberofplayer \totalcontroldim_i$ be the dimension of the joint control input.  We denote by $\state_\stage:=[\state_\stage^i]_{i=1}^\totalnumberofplayer$ and $\control_\stage = [\control_\stage^i]_{i=1}^\totalnumberofplayer$. Let $\initialstatedistribution$ be the initial state distribution. Following standard assumptions in smooth dynamic games involving continuous state and action spaces \cite{bacsar1998dynamic,fridovich2020efficient, laine2023computation}, we consider the differentiable dynamics $f$:
\begin{equation}
    \state_{\stage+1} = \dyn (\state_\stage, \control_\stage^1,\dots,\control_\stage^\totalnumberofplayer), \  \ \state_0\sim \initialstatedistribution,\ \ t = 0,1,\dots
\end{equation} 
and each agent $i$ has an individual differentiable stage cost\footnote{This assumption does not restrict the applicability of our method to games with non-differentiable objectives or dynamics; for instance, in the six-player strategic basketball formation experiments of Section~\ref{sec:results}, our method still outperforms existing MA-PG approaches despite the non-differentiable objective.
} $\stagecost^i(\state_\stage, \control_\stage^1,\dots, \control_\stage^\totalnumberofplayer)$. Let $\policy^i(x) $ be the time-invariant policy of agent $i$, which in general is a random distribution conditioned on the state variable $x$. Define $\policy(\state): = [\policy^1(\state), \dots, \policy^N (\state)]$, and $\control\sim \policy(\state)$, where $\control^i \in \policy^i(\state)$, $\forall i\in[N]$. Denote by $\policy^{-i} : = \policy \setminus \{\policy^i\}$. 
Given the policies $\{\policy^i\}_{i=1}^N$, we define the value function of the $i$-th agent $\valuefunc^{i,\policy}(\cdot): \mathbb{R}^n \to \mathbb{R}$ as 
\vspace{-0.5em}
\begin{equation}
    \valuefunc^{i,\policy}(x) =  \mathbb{E}_{\substack{\control_\stage \sim\policy(\state_\stage)}}\Big[\sum_{\stage=0}^\infty  \gamma^\stage \stagecost^i  (\state_\stage, \control^1_\stage, \cdots  \control^N_\stage)  \Big],
\end{equation}
where $\gamma\in(0,1]$ is the time-discount factor. We have the Bellman equation for the $i$-th agent,
\begin{equation}
    \resizebox{1\hsize}{!}{$\displaystyle\valuefunc^{i,\policy}(\state) = \mathbb{E}_{\substack{\control_\stage \sim \policy(\state_\stage)}} \big[ \stagecost^i(\state, \control^1,\cdots, \control^N )  + \gamma \valuefunc^{i,\policy}(\dyn(\state, \control^1,\dots, \control^N)) \big].$}
    \label{eq:bellman}
\end{equation}

In a non-cooperative game, each agent $i\in\totalnumberofplayer$ minimizes its cumulative cost $\sumcost^i$ over its policy $\policy^i$, while fixing $\policy^{-i}$
\begin{equation}\label{eq:infinite horizon non-cooperative games}
    %\min_{\policy^i}
    \sumcost^i (\policy^i, \policy^{-i}):= \mathbb{E}_{\state_0\sim\initialstatedistribution} \big[ \valuefunc^{i,\policy}(\state_0) \big].
    % \mathbb{E}_{\substack{\state_0\sim\initialstatedistribution \\ \control_\stage \sim \policy(\state_\stage)} } \Big[\sum_{\stage=0 }^\infty  \gamma^\stage \stagecost^i  (\state_\stage, \control_\stage^1, \cdots, \control_t^N)  \Big]
\end{equation}
A set of policies $\{\optimalpolicyofplayeri\}_{i=1}^N$ is a \emph{local feedback Nash equilibrium} if there exists $\epsilon>0$ such that for all $i\in[N]$, $\sumcost^i(\policyofplayeri, \policy^{-i*}) \ge \sumcost^i(\optimalpolicyofplayeri, \policy^{-i*}),\  \forall \policyofplayeri,\  \text{s.t.}\; D_{KL}(\policyofplayeri | \optimalpolicyofplayeri) \le \epsilon,$ 
% \begin{equation*}
%     \sumcost^i(\policyofplayeri, \policy^{-i*}) \ge \sumcost^i(\optimalpolicyofplayeri, \policy^{-i*}), \ \ \forall \policyofplayeri \;\text{s.t.}\; D_{KL}(\policyofplayeri | \optimalpolicyofplayeri) \le \epsilon,
% \end{equation*}
where $D_{KL}$ is the KL divergence. For deterministic policies, this can be replaced with the L2 norm.
 
In MA-PG methods, each agent $i$'s policy is updated as:
\begin{equation}\label{eq:simultaneous gradient play}
    \policy^i \gets \policy^i - \stepsize \nabla_{\policy^i} \sumcost^i(\policy^i, \policy^{-i}),
\end{equation}
where $\stepsize$ is the step size. 
We denote the concatenation of all agents' policy gradient as the \textit{pseudo-gradient} $w(\policy)$:%, given as
\begin{equation}
    w(\policy) : = \begin{bmatrix}
        \nabla_{\policy^{1}} \sumcost^1(\policy^{1}, \policy^{-1}) \\ 
        \vdots \\ 
        \nabla_{\policy^{\totalnumberofplayer}} \sumcost^\totalnumberofplayer(\policy^{\totalnumberofplayer}, \policy^{-\totalnumberofplayer})
    \end{bmatrix} .
\end{equation}
We can represent the MA-PG update as a dynamical system with respect to $\policy$, often referred to as the \emph{MA-PG dynamics}:
\begin{equation}\label{eq:mapg dynamics}
    \policy \gets \policy - \eta  w(\policy).
\end{equation}
% Denote by $j$ the policy gradient iteration number. We can write the multi-agent policy gradient update compactly as
% \begin{equation}
%     \policy^{(j+1)} \gets \policy^{(j)} - \stepsize \cdot  w(\policy^{(j)}).
% \end{equation}
\subsection{Non-cooperative linear quadratic games}
\label{sec:lq}
%In this subsection, we introduce a special class of non-cooperative games, linear quadratic (LQ) games, which are a class of games that has been well-studied \citep{bacsar1998dynamic}. Since LQ games are a relatively simple setup compared to many recent multi-agent deep RL problems \citep{bansal2018emergent}, we mainly use this class of games for characterizing the theoretic property of policy gradient methods when solving non-cooperative games. 
In this subsection, we introduce a special class of non-cooperative games known as linear–quadratic (LQ) games \cite{bacsar1998dynamic}, a simplified setting that allows theoretical characterization of the performance of policy gradient methods. In infinite-horizon LQ games, we consider for each agent $i\in[N]$,
\begin{align}
        \min_{\policy^i}\  & \expectationoverinit \Big[\sum_{\stage=0}^\infty \state_\stage^\top \statecostmatrix^i \state_\stage +  \control_\stage^{i\top} \controlcostmatrix^i \control_\stage^i \Big] \label{eq:lq-game-formulation}\\
        \textrm{s.t. }& \state_{\stage+1} = \statematrix \state_\stage +\sum_{i=1}^N\controlmatrix^i \control_\stage^i, \; \control_\stage^i \sim \policy^i(\state_\stage),\; \forall t=0,1,\dots \label{eq:time-invariant linear dynamics} \nonumber
\end{align}
where $A\in\mathbb{R}^{n\times n}$, $B^i\in \mathbb{R}^{n\times m_i}$, $Q^i\in\mathbb{R}^{n\times n}$ and $R^i \in\mathbb{R}^{m_i\times m_i}$, $\forall i\in[N]$. We consider the following assumption:
% \vspace{-0.5em}

% We derive our theoretical analysis under the following assumption.
\begin{assumption}\label{assumption: infinite horizon LQ}
    The matrices $\{\statecostmatrix^i, \controlcostmatrix^i \}_{i=1}^N$ are positive definite. The linear dynamics in \eqref{eq:lq-game-formulation} is stabilizable, i.e., there exists a set of time-invariant feedback matrices $\{\linearpolicy^i\}_{i=1}^N$ such that the closed-loop dynamics {$\state_{\stage+1} = \bar{A}_{K}\state_\stage  $ is exponentially stable, where $\bar{A}_{K}:=\statematrix - \sum_{i=1}^N \controlmatrix^i \linearpolicy^i$}. %In other words, all eigenvalues of the closed-loop dynamics matrix $\closedloopstatematrix:=\statematrix - \sum_{i=1}^N \controlmatrix^i \linearpolicy^i$ has a magnitude strictly less than $1$. 
    To ensure sufficient exploration, the covariance matrix $\Sigma_0:=\mathbb{E}_{x_0}[x_0x_0^\top]$ of the initial state distribution $\initialstatedistribution$ is assumed to be full rank. 
\end{assumption}
\vspace{-0.5em}
Assumption~\ref{assumption: infinite horizon LQ} is used only for the theoretical analysis of our method and is not required in practical experiments. It ensures the existence of a Nash equilibrium of \eqref{eq:lq-game-formulation}. In LQ games with deterministic dynamics, we can specify the agent's policy as a deterministic linear feedback policy, each given by $u_t^i = -K^i x_t$. With a slight abuse of notation $w(\cdot)$, we consider the following pseudo-gradient \cite{mazumdar2020policy}
% \end{assumption}
 % In what follows, we review the basics of policy gradient method for computing the feedback Nash equilibrium. Define the \textit{pseudo-gradient mapping function} \citep{mazumdar2020policy}
\begin{equation}
    w(K):= \begin{bmatrix}
        \nabla_{K^1} \sumcost^1(K^1,\dots,K^N)\\
        \vdots \\
        \nabla_{K^N} \sumcost^N(K^1, \dots, K^N)
    \end{bmatrix}, K:=[K^i]_{i=1}^N, 
    \label{eq:gradient_mapping}
\end{equation}
to analyze the policy gradient method for computing the feedback Nash equilibrium.
We analytically derive the policy gradient by substituting the linear policy into the Bellman equation \eqref{eq:bellman},
\begin{align}
    \lqsumcost^i( & K^1,\dots,K^N) =  \label{eq:derive_condition_K1} \\ 
    & \expectationoverinit \Big[ x_0^\top Q^i x_0 + x_0^\top {K^i}^\top R^i K^i x_0 + x_0^\top \bar{A}_{K}^\top P_K^i \bar{A}_{K} x_0 \Big], \nonumber
\end{align}
where $P_{K}^{i}$ is given as the unique positive definite solution of the linear quadratic Bellman equation
\begin{equation}
 P_{K}^{i} = \bar{A}_{K}^{\top} P_{K}^{i} \bar{A}_{K} + {K^{i}}^{\top} R^i K^{i} + Q^i.
\label{eq:P-bellman}\vspace{-0.25em}
\end{equation}
Under Assumption~\ref{assumption: infinite horizon LQ}, we have $\mathbb{E}_{x_0}\big[x_0x_0^\top\big]\succ 0$ and $\Sigma_K := \mathbb{E}_{x_0} \big[ \sum_{t=0}^{\infty} x_t x_t^{\top}\big] \succ 0$. As noted in \cite{mazumdar2020policy}, by taking the gradient of \eqref{eq:derive_condition_K1} with respect to $K^i$, we have
\begin{equation}
    \nabla_{K^i} \lqsumcost^i(K^1,\dots,K^N) = 2 (R^i K^i - {B^i}^{\top} P_{K}^{i}  \bar{A}_{K} ) \Sigma_{K}.
\label{eq:derive_condition_K2}    
\end{equation}

\section{Analysis: Stability and bias of the guided policy optimization in LQ games}\label{sec:lq-results}
% \section{Analysis: Linear Quadratic Games}
In this section, we present our analytical results for guided policy optimization in LQ games. First, we introduce a model-based guidance term into the original cost function, % of the game defined in \eqref{eq:infinite horizon non-cooperative games}, 
which serves as a regularization mechanism. We then focus specifically on the theoretical analysis of LQ games.% in Section \ref{sec:lq-results}.

% \subsection{Stability and bias of the guided policy gradient}
% \label{sec:lq-results}

To formalize our guided policy optimization in  LQ games, let $\check{K}$ be an arbitrary stabilizing feedback policy for the dynamics in \eqref{eq:lq-game-formulation}. % and may potentially differ from the Nash equilibrium policy $K^*$.
%, which could be derived by using the first-stage policy of a finite horizon feedback Nash equilibrium policy that can be easily computed. 
We consider the LQ games in Section \ref{sec:lq}, but now with a cost for agent $i$, regulated by a deviation from the guiding control $\check{K}^i x_t$, given as
\begin{align}
    \sumcostregularized^i & ( K^1,\dots,K^N) \label{eq:guided_lq_cost} \\ 
    {\scriptstyle :=} & \resizebox{0.96\hsize}{!}{$\displaystyle \mathbb{E}_{x_0}\!\Big[\! \sum_{t=0}^\infty \big(x_t^\top Q^i x_t + \state_\stage^\top \linearpolicy^{i\top} R^i \linearpolicy^i \state_\stage \big) +\rho {x_t^\top (K^i - \check{K}^i)^\top R^{i} (K^i - \check{K}^i) x_t} \Big]$}  \nonumber \\
    {\scriptstyle =}&  \resizebox{0.8\hsize}{!}{$\displaystyle\lqsumcost^i{(K^1,\dots,K^N)} +\rho \mathbb{E}_{x_0}\Big[ \sum_{t=0}^\infty {x_t^\top (K^i - \check{K}^i)^\top R^{i} (K^i - \check{K}^i) x_t}\Big].$} \nonumber
\end{align}
We denote $\hat{K}^*=[\hat{K}^{1*},\dots, \hat{K}^{N*}]$ as the Nash equilibrium solution for this \textit{guiding-policy-regularized LQ game}.
% \begin{equation}
%     \hat{K}^i:=\arg\min_{K^i} \mathbb{E}_{x_0}\Big[ \sum_{t=0}^\infty \frac{1}{\rho}\big(x_t^\top Q^i x_t + u_t^\top R^i u_t \big) + \|(K^i - \check{K}^i)x_t\|_2^2\Big]
% \end{equation}
We show that the introduced regularization term can resolve the instability of the MA-PG dynamics~\eqref{eq:mapg dynamics}{\color{black}, thereby guaranteeing local exponential convergence to a stationary point of the regularized game}.

\begin{proposition}\label{prop: L2 stabilization}
    % Consider a linear quadratic game. 
    Under Assumption~\ref{alg:ma-gps}, let $K^*$ be a feedback Nash equilibrium policy, $\check{K}$ an arbitrary stabilizing feedback policy, and $\mathcal{K}$ the set of all stabilizing feedback policies. Then, we have
    \begin{enumerate}
        \item $w(K)$ is locally Lipschitz continuous on $\mathcal{K}$. There exists a radius $r$ such that $\forall K \in \mathcal{K}_r:= \{ K \in \mathcal{K} : \| K - K^* \|_2^2 \leq r \}$, the real parts of the eigenvalues of the Jacobian matrix $\nabla w(K)$ are lower-bounded by some~$\lambda_{\min}$.
        \item As $\rho$ increases, the regularized Nash equilibrium $\hat K^*$ will exhibit locally stable policy gradient dynamics, with all eigenvalues of the Jacobian of the regularized pseudo-gradient $\big(\nabla w(\hat K^*) + \rho\, \Gamma(\hat K^*)\big)$ having positive real parts, where $\Gamma$ denotes the Jacobian of the pseudo-gradient of the guidance term in (13), which satisfies
        \begin{equation*}
            \Gamma(\check K) = 2\begin{bmatrix}
\Sigma_{\check K}\otimes R^1 & & \\
 & \ddots & \\
 & & \Sigma_{\check K}\otimes R^N
\end{bmatrix} \succ 0.
        \end{equation*}
        % \item %Suppose that the guidance policy satisfies $\|\check{K}^i - K^{i*}\|_2\le \frac{1}{2N}r$, $\forall i\in \{1,2,\dots,N\}$. 
        % As $\rho$ increases, the regularized Nash equilibrium $\hat{K}^*$ will exhibit locally stable policy gradient dynamics, with all eigenvalues of the regularized pseudo-gradient $(\nabla w(K) + \rho \cdot R)$ having positive real parts. %Specifically, policy gradient methods locally converge exponentially fast to $\hat{K}^*$. 
    \item The bias of the regularized Nash equilibrium $\hat{K}^{i*}$ of each agent $i$ {is characterized by}
    \begin{align}
        \Big\|&\hat{K}^{i*}  - \left(\frac{\rho}{1+\rho}\check{K}^i + \frac{1}{1+\rho}K^{i*}\right)\Big\|_2 \label{eq:bias-bound} \\
        & \le \frac{1}{1 + \rho}\|{R^i}^{-1}\|_2 \|B^{i\top}\|_2 \big\|\hat{P}_{\hat{K}^*}^i \bar{A}_{\hat{K}^*}-P_{K^{*}}^{i}\bar{A}_{K^*}\big\|_2, \nonumber
    \end{align}
where $\bar{A}_{\hat{K}} = A - \sum_{i=1}^{N} B^i \hat{K}^i$, and $\hat{P}_{\hat{K}}^i$ is given as the unique positive definite solution of $\hat{P}_{\hat{K}}^i = {\bar{A}_{\hat{K}}}^{\top} \hat{P}_{\hat{K}}^i  \bar{A}_{\hat{K}} + \hat{K}^i {}^\top R^i \hat{K}^{i} + Q^i + \rho {(\hat{K}^i - \check{K}^i)}^\top R^i (\hat{K}^i - \check{K}^i)$.%\vspace{-0.3em}
% \begin{equation}
% \hat{P}_{\hat{K}}^i = {\bar{A}_{\hat{K}}}^{\top} \hat{P}_{\hat{K}}^i  \bar{A}_{\hat{K}} + \hat{K}^i {}^\top R^i \hat{K}^{i} + Q^i + \rho {(\hat{K}^i - \check{K}^i)}^\top R^i (\hat{K}^i - \check{K}^i).\vspace{-0.3em}
% \label{eq:P-bellman-hat}
% \end{equation}
    \end{enumerate} %Denote by $e_i$ the $i$-th column of the $n$-dimensional identity matrix. Let the initial state distribution be a uniform distribution such that the probability of having $x_0 = e_i$ is $\frac{1}{n}$.     
    % \jcnote{Suppose}  %under the following conditions:
    % \begin{enumerate}
        %\item There exists a radius $r$ such that, for all stabilizing control policies $K \in \mathcal{K}_r:= \{ K \in \mathcal{K} : \| K - K^* \|_2^2 \leq r \}$, there exists a lower bound $\lambda_{\min}$ for the real parts of the eigenvalues of $\nabla w(K)$.
        %There exists a radius $r$ such that there exists a lower bound $\lambda_{\min}$ for the real parts of the eigenvalues of $\nabla w(K)$ for all stabilizing control policy $K\in\{K\in \mathcal{K}: \|K - K^*\|_2\le r\}$. 
        % \item The guidance linear policy satisfies $\|\check{K}^i - K^{i*}\|_2\le \frac{1}{2N}r$, for each $i\in \{1,2,\dots,N\}$.
    % \end{enumerate}
    % \begin{enumerate}
    %     \item the regularized Nash equilibrium $\hat{K}^*$ is in a $r$-neighborhood of $\check{K}$.
    %     \item there exists an open set around $\check{K}$ such that the Jacobian of $w(K)$ has all positive eigenvalues. In other words, the policy gradient dynamics is locally stable around the regularized Nash equilibrium $\hat{K}^*$ within this open set. 
    % \end{enumerate}
\end{proposition}
\begin{proof}
    We include the proof in Appendix.
    % Appendix \ref{appendix:proof of stable policy gradient dynamics}. 
\end{proof}
\begin{figure*}[t!]
    \centering
    \includegraphics[width=1\linewidth]{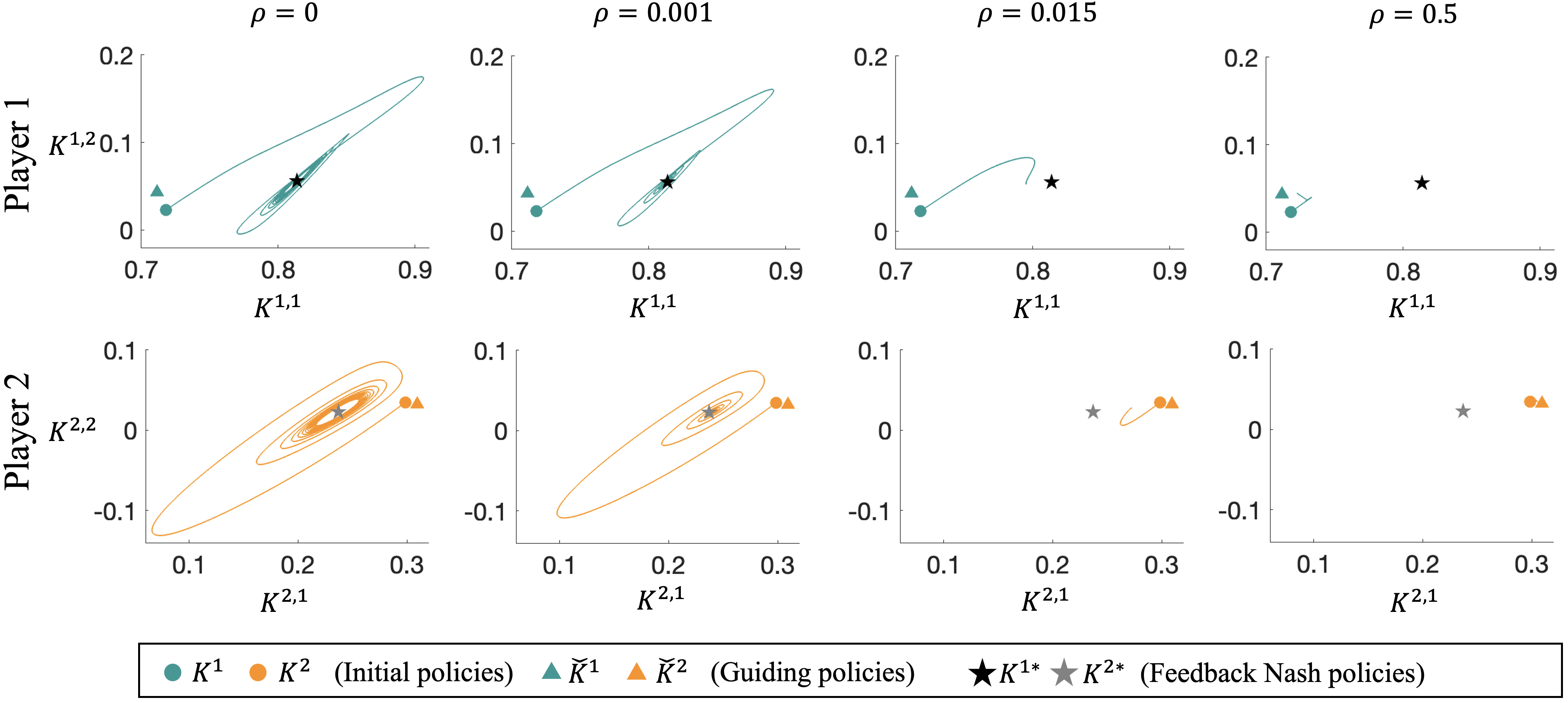}
    \vspace{-1.5em}
    \caption{\textbf{Policy guidance stabilizes infinite-horizon LQ game dynamics, enables convergence to a close neighborhood of the true Nash equilibrium—even if the guidance is wrong.%Imperfect policy guidance stabilizes the policy gradient dynamics in infinite-horizon LQ games, even under a wrong guiding policy. 
    }
    Each row corresponds to one of $K^{1}$ and $K^{2}$. In each plot, the two axes represent the entries $K^{i,1}$ and $K^{i,2}$ of the corresponding $K^{i}$, and the curves depict their trajectories under policy gradient with different values of $\rho$ and a biased guiding policy $\check{K}$ (potentially derived from an inaccurate dynamics model). %The third row shows the trajectories of the values of each entry of $K^1, K^2$ over policy gradient steps.  
    A small $\rho$ ($\ge 0.001$) is sufficient to stabilize the policy gradient dynamics near the ground truth feedback Nash equilibrium  $K^{*}$. However, as $\rho$ increases, the bias in the converged policies introduced by the guidance also grows. }
    \label{fig:lq-plot}
    \vspace{-1em}
\end{figure*}

Proposition \ref{prop: L2 stabilization} indicates that adding the guidance term from \eqref{eq:guided_lq_cost} with a sufficiently large $\rho$ stabilizes the policy gradient dynamics, guaranteeing the local exponential convergence of policy gradient methods to a stationary point even if the original policy gradient dynamics is unstable (i.e., $\lambda_{\min}<0$). 
This stability is achieved because, as $\hat K^*$ approaches $\check K$, $\Gamma(\hat K^*)$ approaches the positive definite matrix $\Gamma(\check K)$, so that the term $\rho\,\Gamma(\hat K^*)$ dominates the Jacobian $\nabla w(\hat K^*)+\rho\,\Gamma(\hat K^*)$ of the regularized pseudo-gradient and shifts all its eigenvalues to the positive domain.
% Denote by $R\in \mathbb{R}^{m\times m}$ a diagonal block matrix where the $i$-th diagonal block is composed by $R^i$. This stability is achieved because the term $(\rho \cdot R)$ shifts all eigenvalues of the regularized pseudo-gradient's Jacobian $(\nabla w(K) + \rho \cdot R)$ to the positive domain. 
However, the distance (bias) between the converged solution and the true Nash equilibrium depends jointly on the problem dynamics, $\rho$, and the quality of the guidance policy. We include a numerical study in Figure~\ref{fig:lq-plot} to support this. 
%Proposition~\ref{prop: L2 stabilization} also suggests that both guidance policy and $\rho$ contribute to the bias of the converged solution under the guided policy gradient, but a perfect guidance policy is not required to reduce bias under guided policy gradients. 
%We include a numerical study in Figure~\ref{fig:lq-plot} showing that the solution converged by guided policy gradients lies between the true Nash equilibrium and the guidance policy, controlled by the parameter $\rho$. 
Specifically, there are three cases: 
\begin{enumerate}
    \item When $\rho$ is small, we may not be able to stabilize the limit cycle behavior of the policy gradient dynamics if $(\nabla w(K) + \rho \cdot \Gamma(K))$ has an eigenvalue with a negative real part;
    \item Increasing $\rho$ to an appropriately large value, we enhance the stability of the policy gradient dynamics, and the converged solution under the guided policy gradient dynamics could be close to the Nash equilibrium of the original game;
    \item Further increasing $\rho$, we will enhance stability, but the converged solution may deviate from the Nash equilibrium of the original game and be biased more towards the guidance policy.
\end{enumerate}
These observations suggest that $\rho$ trades off stability and bias. {\color{black}Moreover,} if the guidance matches the true Nash equilibrium, additional policy gradient updates become unnecessary; otherwise, the converged solution under the guided policy gradient dynamics can outperform an imperfect guidance policy $ \check{K}$.% with local exponential convergence. %If the guidance policy matches the ground truth Nash equilibrium, policy gradients become unnecessary. However, if the guidance is imperfect, our method allows learning policies that outperform the given guidance.

\section{Guiding Multi-Agent Policy Gradient in Nonlinear Games}
\label{section:main_method}
In this section, we extend the idea of leveraging a model-based prior to stabilize policy gradients in LQ games to the context of nonlinear games. We introduce an efficient multi-agent guided policy search algorithm that incorporates this model-based prior for nonlinear game scenarios.

\label{sec:iLQGames}

% * define rollout trajectory\\
% * describe the LQ approximation \\
% * construct a local feedback controller 
\subsection{Efficient computation of local guidance for MA-PG}
\label{sec:local guidance}
Unlike the LQ scenario, where a stabilizing policy can be easily computed as guidance, it is generally hard to obtain guidance for MA-PG for nonlinear games due to the complexity of dynamics. Inspired by the iLQGames method \cite{fridovich2020efficient}, which argues that along the Nash equilibrium state and control trajectories of a nonlinear game, the Nash equilibrium of the locally constructed LQ game along those state and control trajectories can approximate the Nash equilibrium of the original nonlinear game well. Building upon this idea, we construct local LQ game approximations along the neural network control output and the resulting state trajectory, and then use the Nash equilibrium of those local LQ game approximations as reference signal to guide the learning of multi-agent policy gradients. 

The key innovation of our approach is efficiently utilizing such a local guidance, rather than repeatedly computing iLQGames solutions until convergence at each sampled state. This guidance aggregates multi-agent gradient information, aligning policy updates across strategic agents toward a consistent Nash equilibrium and thereby stabilizing and accelerating MA-PG in settings with multiple equilibria arising from complex objectives.

To achieve this, we introduce the time-varying finite-horizon linear quadratic games, with a given initial state $\nominalstate_0$ and $i\in[N]$:\vspace{-0.5em}
\begin{equation}
\begin{aligned}\label{eq:finite horizon LQ games}
    \min_{\{x_t\}_{t=0}^{T+1},\{\control_\stage^i\}_{t=0}^{T}} \ & \sum_{\stage=0}^T \state_\stage^\top \statecostmatrix_\stage^i \state_\stage  + 2\statecostlinearterm_t^{i\top} \state_\stage +  \control_\stage^{i\top} \controlcostmatrix_\stage^i \control_\stage^i + 2\controlcostlinearterm_\stage^{i\top} \control_\stage^i \\%+ \frac{1}{2} x_{T+1}^\top Q^i_{T+1} x_{T+1} + q_{T+1}^{i\top} x_{T+1} \\
    \textrm{s.t. }& \state_{\stage+1} = \statematrix_\stage \state_\stage + \sum_{i=1}^N \controlmatrix_\stage^i \control_\stage^i %+ \dynoffset_\stage^i
    ,\ \state_0 = \nominalstate_0, \\
    & \forall t=0,\dots,T
\end{aligned}
\end{equation}
%We note that the horizon length $T$ is a user-defined parameter, and if computational resources are constrained, $T$ may be selected to be relatively short. We find $T=10$ or $20$ is sufficiently long for generating reasonable guidance control for our examples in Section~\ref{sec:results}, but this would depend on the dynamics and the cost of the game. In this work, we propose to learn a neural network policy $\NNpolicy$ for the nonlinear games using model-prior guided policy gradients. In what follows, we describe the procedure of computing a reference trajectory for computing the policy guidance. 
The horizon length $T$ is user-defined and can be relatively short if computational resources are limited. In our experiments (Section~\ref{sec:results}), we find $T=10$ or $20$ sufficient for generating effective guidance, though the optimal choice depends on game-specific dynamics and costs. We propose training a neural network policy \( \NNpolicy \) for nonlinear games via model-based guided policy gradients. Below, we describe the procedure for computing the reference trajectory used to obtain policy guidance. Consider a $T$-horizon nominal state trajectory $\{\nominalstate_\stage\}_{\stage=0}^{T+1}$ steered by the \emph{neural network control} $\{\nominalcontrol_\stage\}_{\stage=0}^T$, with $\nominalcontrol_\stage  \sim \NNpolicy(\nominalstate_\stage)$, $\nominalstate_{\stage+1} = \dyn(\nominalstate_\stage, \nominalcontrol_\stage^1,\dots,\nominalcontrol_\stage^N), \ \ \forall t = 0,\dots, T$.
We linearize the dynamics and quadraticize the cost to derive a local LQ game approximation around the nominal trajectory $\{\nominalstate_0,\nominalcontrol_0,\cdots,  \nominalstate_{T+1}\}$, for all $i\in[N]$,
\begin{equation*}{
    \begin{aligned}
    &\statematrix_\stage : =  \nabla_{\nominalstate_\stage} \dyn (\nominalstate_\stage, \nominalcontrol_\stage), 
    &&\controlmatrix_\stage^i : = \nabla_{\nominalcontrol^i_\stage} \dyn(\nominalstate_\stage, \nominalcontrol_\stage), \\ 
    &\statecostmatrix_\stage^i: =  \nabla^2_{\nominalstate_\stage} \stagecost^i(\nominalstate_\stage, \nominalcontrol_\stage), 
    &&q_\stage^i: = \nabla_{\nominalstate_\stage} \stagecost^i(\nominalstate_\stage,\nominalcontrol_\stage), \\
    & \controlcostmatrix_\stage^i:=  \nabla^2_{\nominalcontrol^i_\stage} \stagecost^i(\nominalstate_\stage, \nominalcontrol_\stage),  &&r_\stage^i : = \nabla_{\nominalcontrol_\stage^i} \stagecost^i(\nominalstate_\stage,\nominalcontrol_\stage).
    \end{aligned}}
\end{equation*}

% Computing long-horizon Nash equilibria using iLQGames can be time-consuming for efficient integration into deep RL training. Therefore, instead of running iLQGames until convergence for every state $\bar{x}_0$ sampled during training, our method directly uses the \emph{first} element of the Nash equilibrium control $\{u_t^*\}_{t=0}^T$ of the LQ game approximation \eqref{eq:finite horizon LQ games} to define the guidance control of $\nominalstate_0$. Moreover, when multiple Nash equilibria exist for the local LQ games, enumeration of all Nash equilibria is \emph{unnecessary} and we can use the minimum L2-norm Nash equilibrium \cite{laine2023computation} as a synchronized guidance to ensure all decentralized agents policies are consistently regulated towards a synchronized Nash equilibrium. 

Rather than running iLQGames to convergence for every sampled state $\bar{x}_0$ during training, our method directly uses the \emph{first} element of the Nash equilibrium control sequence $\{u_t^*\}_{t=0}^T$ from the local LQ approximation \eqref{eq:finite horizon LQ games} as the guidance control for $\nominalstate_0$. Since the Nash equilibrium control $\{u_\stage^*\}_{t=0}^T$ of local LQ game \eqref{eq:finite horizon LQ games} is defined in the local coordinate around the nominal control trajectory $\{\nominalcontrol_\stage\}_{\stage=0}^T$, the guidance control $\check{u}_0$ of $\nominalstate_0$, which is defined in the global coordinate, should be defined as $ \check{u}_0 :=\nominalcontrol_0 + u_0^*$. For agent $i\in[N]$, we define the mapping from $(\nominalstate_0, \NNpolicy)$ to the guidance control $\check{u}_0^i$ as the guidance policy $\check{\policy}^i: (\nominalstate_0, \NNpolicy)\to \check{u}^i_0$. With the model-based prior $\check{\pi}^i$, we introduce the guided policy stage cost of agent~$i$:\vspace{-0.3em}
\begin{equation}\label{eq:modified reward}
    \hat{\stagecost}_{\rho}^i (\state, u^1,\dots,u^N) := \stagecost^i(\state,u^1,\cdots, u^N) + \rho \|u^i - \check{\policy}^i(x, \NNpolicy)\|_2^2\vspace{-0.3em}
    % \min_{\NNpolicy^i}  \sumcost^i(\NNpolicy^i, \NNpolicy^{-i}) + \mathbb{E}_{x\in\mathbb{P}} \| \NNpolicy^i(\state) - \guidancecontrol^i \|_2^2
\end{equation}
where $\rho>0$ is the guidance weight. 

\begin{remark}[On the uniqueness of Nash equilibria]
When multiple Nash equilibria exist in local LQ games, enumerating all Nash equilibria is unnecessary and computationally excessive. Instead, we adopt the unique \emph{minimum L2-norm Nash equilibrium} \cite{laine2023computation} as synchronized guidance across agents. Proposition~\ref{prop: L2 stabilization} shows that such synchronized guidance can effectively stabilize MA-PG dynamics. Empirically, we find that this guidance enables our method to outperform prior MARL approaches and converge faster in these settings, as shown in Section~\ref{sec:results}.
\end{remark}

\begin{algorithm}[t!]
% \SetAlgoLined
\caption{Multi-agent guided policy search}
\label{alg:ma-gps}
\vspace{0.5em}
\KwIn{Neural network policies of each player $\{\pi_\theta^i\}_{i=1}^N$ and the neural network value functions $\{V_{\theta'}^{i,\pi_\theta}\}_{i=1}^N$, %policy guidance weight $\rho > 0$, 
learning rate $\eta$, list of guidance weights $\{\rho_{i}\}_{i=1}^{N_e}$, 
number of epochs $N_e$, number of gradient steps $N_g$, number of sampled rollout trajectories $M_r$, batch size $M_b$, rollout trajectory length $T_r$} %, local guidance horizon $T_g$.}
\KwOut{Optimized policies $\{\pi_\theta^i\}_{i=1}^N$}

\hspace*{-0.5em} Set data buffer $D = \{\}$\;
\For{$k \gets 1$ \KwTo $N_e$}{
    % \tcc{Perform testing for current epoch}
    % \texttt{test1}\;
    % \texttt{test2}\;
    % \tcc{Sample data and perform gradient updates}
    % $\{x_{ij}^N\}_{i=1,j=1}^{N,M}, \{u_{ij}\}_{i=1,j=1}^{N,M} \gets$ Sample uniformly across $S$ and $U$\;
    % $\{y_{ij}\}_{i=1,j=1}^{N,M} \gets \texttt{LABEL}(\{x_{ij}^N\}_{i=1,j=1}^{N,M}, \{u_{ij}\}_{i=1,j=1}^{N,M}, d)$\;
    \For{$j \gets 1$ \KwTo $N_g$}{
        Sample $M_r$ initial conditions $\{x_0^{(j)}\}_{j=1}^{M_r}$. \\
        For each $j\in\{1,\dots, M_r\}$, 
        simulate a trajectory $\xi^{(j)}=\{x_{0}^{(j)}\!, u_0^{(j)}\!,
        \dots, x_{T_r+1}^{(j)}\}$, under the policies $\{\pi_\theta^i\}_{i=1}^N$. \\
        Update data buffer, $D\gets \{\xi^{(j)}\}_{j=1}^{M_r}$\\
        \For{each agent $i\in[N]$}{
            Sampling a batch of states $\{x^{(j)}\}_{j=1}^{M_b}$\\
            %Obtain the guidance $\{\check{u}^{(j)}\}_{j=1}^{M_b}$ for each $\{x^{(j)}\}_{j=1}^{M_b}$, as outlined in Section~\ref{sec:local guidance}\\
            Optimizing policy: $\theta \gets \theta - \stepsize \nabla_{\theta}  \sum_{j=1}^{M_b} \Big[\hat{\stagecost}_{\rho_k}^i \big(x^{(j) },\pi_\theta(x^{(j)})\big)+\gamma V_{\theta'}^{i,\pi_\theta}\Big(f\big(x^{(j) },\pi_\theta(x^{(j)})\big)\Big)\Big] $ \;%\hat{L}_{i}^{PPO}(\theta, \rho_k )$\;
            Optimizing value function via minimizing Bellman error \cite{sutton2018reinforcement}:
            $\theta' \gets \theta' - \stepsize \nabla_{\theta'} \sum_{j=1}^{M_b} \Big\| V_{\theta'}^{i,\pi_\theta}(x^{(j)}) -  \hat{\stagecost}_{\rho_k}^i\big(x^{(j)}, \pi_\theta(x^{(j)})\big)-\gamma V_{\theta'}^{i,\pi_\theta}\Big(f\big(x^{(j)}, \pi_\theta(x^{(j)})\big)\Big)  \Big\|_2 $
        }
    }
}
\end{algorithm}

\subsection{Multi-agent guided policy search for nonlinear non-cooperative dynamic games}
\label{sec:ma-gps}

Leveraging the policy guided cost function \eqref{eq:modified reward}, we introduce our multi-agent guided policy search (MA-GPS) method in Algorithm~\ref{alg:ma-gps}. 
At each iteration of policy gradient, we first sample a number of initial states, simulate nominal trajectories and store them into the data buffer. Subsequently, for efficiently training, we sample a small batch of data from the data buffer, compute their guidance control, and then take one step policy gradient descent for each agent to improve the policy and value function estimations. We iterate until convergence. In practice, we could decay the guidance weight $\rho$ periodically. This allows MA-GPS to tolerate imperfect guidance and unleash the computational power of MA-PG to refine the sub-optimal policy. 

\section{Experiments}
\label{sec:results}
We evaluate our method against MA-PG methods in three examples: \textbf{1)} a two-agent linear quadratic (LQ) game, \textbf{2)} a platooning scenario involving three vehicles merging into a single lane, and \textbf{3)} a 24-dimensional six-player basketball formation, where agents strategically move to desired positions based on their roles. Although the last two examples fall outside Assumption 1, our method remains effective, suggesting that the stability analysis of Theorem 1 may extend beyond the LQ setting. 
%Experiment implementation details\footnote{Code is available at \url{https://github.com/qugechen/MA-GPS}} are provided in the full version of our manuscript~\cite{fullversion}.
Implementation details\footnote{Code is available at \url{https://github.com/qugechen/MA-GPS}} of experiments are included in Appendix.

In our deep RL experiment, we compare our method with IPPO~\cite{de2020independent}, IPPO with L2 regularization (equivalent to zero-input guidance for pure variance reduction), MA-PPO~\cite{yu2022surprising} and MA-DDPG~\cite{lowe2017multi}. Our IPPO implementation is adapted for non-cooperative dynamic games, where each agent learns its own value function and policy based on its individual objective, unlike the cooperative setting in \cite{de2020independent}. Concretely, each agent is equipped with a value and a policy network, both updated via simultaneous gradient play, as outlined in \eqref{eq:simultaneous gradient play}. IPPO and MA-PPO both employ entropy-driven exploration, making them competitive baselines, whereas MA-DDPG learns deterministic policies without it.
% We alse investigate iLQGames computastion time on two tasks. 
% For three-vehicle platooning, the average time is $\mu=0.0485\pm 0.00420$\text{s}.
% For the six-player basketball formation, the average time is $\mu=46.81 \pm\ 4.18$\text{s}.
%  When the number of players and the complexity of the objectives increase, even under the best effort parameter tuning, it is challenging for iLQGames to do real-time decision-making. 
% MA-GPS shifts computation to learning and yields a policy that executes with a single forward pass, suitable for time-critical strategic control.
% We measure wall-clock time for 20-step iLQGames on two tasks: three-vehicle platooning (\(\mu=0.0485\pm0.0042\,\mathrm{s}\)) and six-player basketball formation (\(\mu=46.81\pm4.18\,\mathrm{s}\)). As player count and objective complexity grow, iLQGames becomes impractical for real-time decisions. By contrast, MA-GPS shifts computation to training, enabling time-critical control.
% (see Appendix~\ref{sec:basketball-details} for details).

% H1: Our method converges faster than pure MARL algorithm
% H2: Our method can tolerate the model mismatch.
% H3: our method can trade-off between variance reduction and induced bias
\begin{figure*}[t!]
    \centering
    \includegraphics[width=\textwidth]{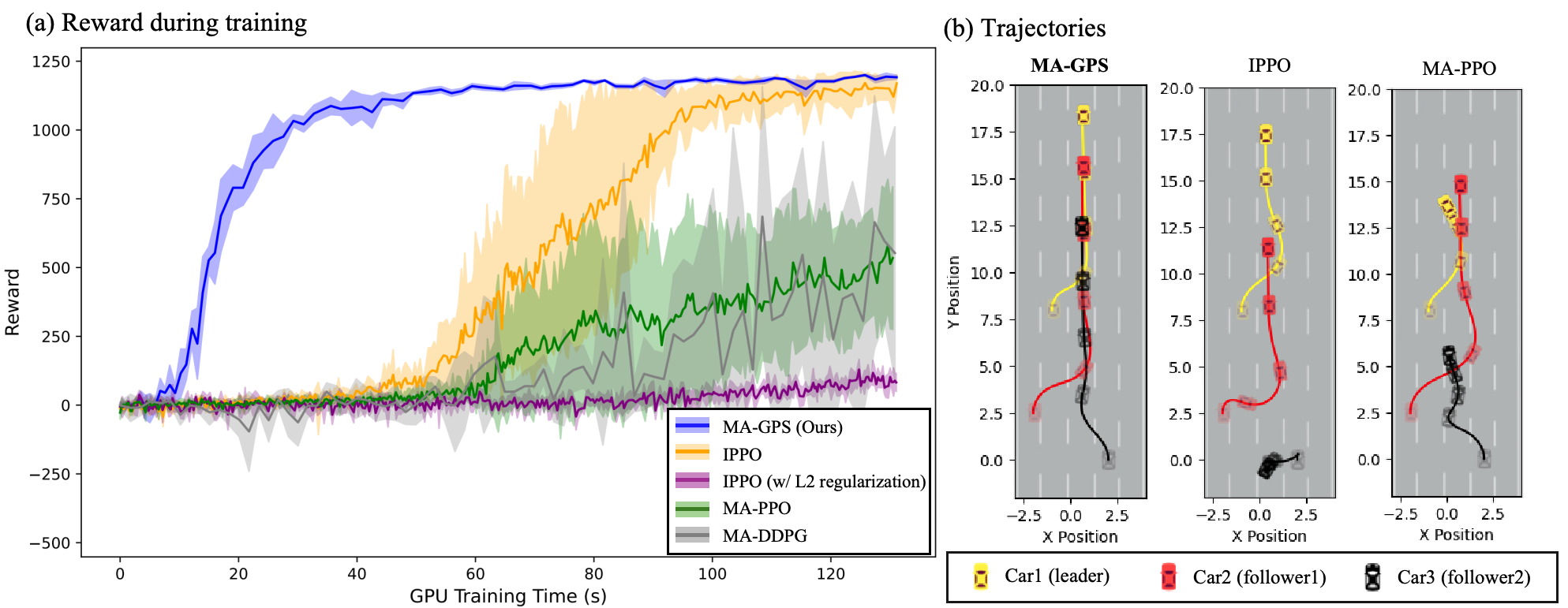}
    \vspace{-1em}
    \caption{\textbf{Three-vehicle platooning experiment.} \textbf{(a) Total reward during GPU training time.} Note that the training time includes the computation time of the LQ game solutions used to guide the MA-GPS. MA-GPS achieves the highest reward and lowest variance compared to the other four methods. \textbf{(b) Trajectories of the vehicles.} The leader vehicle guides the other vehicles to the center lane. MA-GPS effectively enables merging into a lane, while IPPO’s Car 3 gets stuck within the same GPU computation time.}\vspace{-1em}
    \label{fig:three-vehicle}
\end{figure*}
\begin{figure*}[t!]
    \centering
    % \vspace{-1.25em}
    \includegraphics[width=\textwidth]{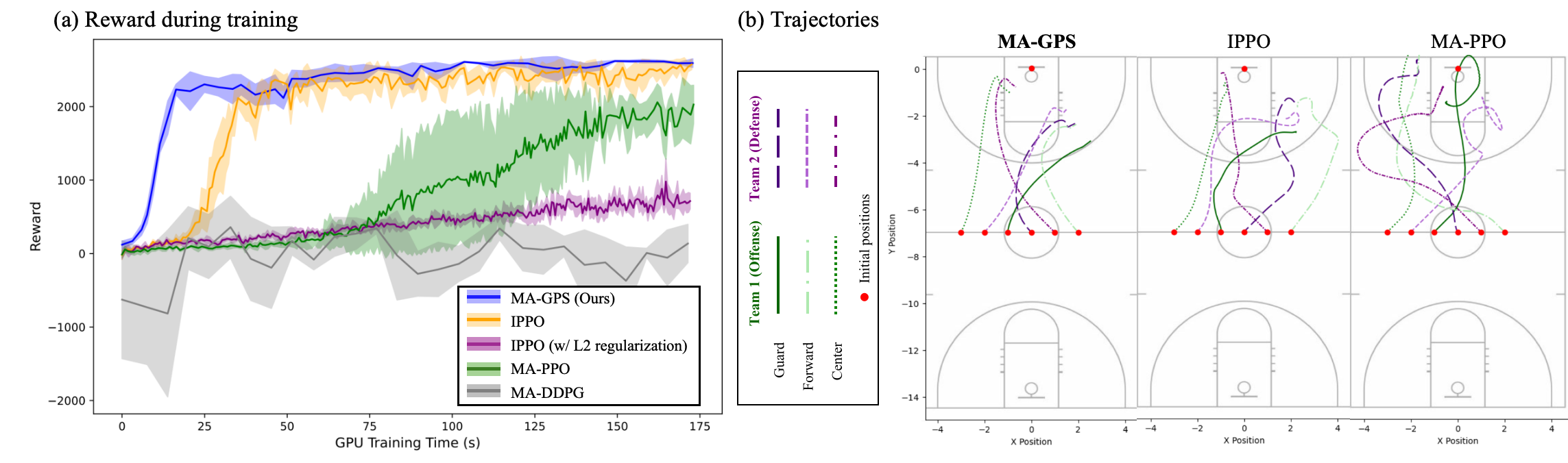}
    % \vspace{-1em} 
    \caption{\textbf{Six-player strategic basketball formation experiment. (a) Total reward during GPU training time.} 
MA-GPS achieves a higher-performing policy more quickly than IPPO, MA-DDPG, and MA-PPO, with lower variance when measured by GPU training time. In contrast, naive L2 policy regularization reduces variance but slows IPPO’s convergence. These results suggest that the local LQ game approximations introduced in Section~\ref{sec:local guidance} can be computed efficiently, thereby accelerating MARL policy convergence.
\textbf{(b) Trajectories of Players.} MA-GPS learns more effective policies than IPPO and MA-PPO within the same GPU computation time, as evidenced by the distinct formation of agents. \vspace{-1em}}
    \label{fig:basketball_all_plot}
    % \vspace{-.8em}
\end{figure*}

\noindent\textbf{1) LQ game}. As noted in \cite{mazumdar2020policy}, MA-PG methods in LQ games exhibit limit-cycle behaviors and fail to converge to Nash equilibria when the Jacobian of $w(K)$ in \eqref{eq:gradient_mapping} has both positive and negative eigenvalues. We reproduce this phenomenon in an LQ game example, detailed in the Appendix, and analyze the stabilizing effect of the proposed guidance regularization term. We solve for the ground truth Nash equilibrium policies $K^{1*}, K^{2*}$ via the Lyapunov iteration from \cite{mazumdar2020policy}. We introduce an artificial bias to the guiding policies $\check{K}^{1}, \check{K}^{2}$ in the regularized cost \eqref{eq:guided_lq_cost}. The resulting policy trajectories under varying regularization weights $\rho$ are shown in Figure~\ref{fig:lq-plot}.

Without guidance ($\rho=0$), policy gradients exhibit limit-cycle behavior and fail to converge. Introducing a small regularization weight $\rho$ effectively shifts all eigenvalues of $\nabla w(K)$ to positive real parts, stabilizing policy gradient dynamics, thus empirically validating Proposition~\ref{prop: L2 stabilization}. As illustrated in Figure \ref{fig:lq-plot}, both agents' policies converge to a neighborhood around their true Nash equilibria, despite imperfect guidance. Increasing $\rho$, however, enlarges the bias in converged policies relative to the ground truth. Hence, the optimal regularization weight $\rho$ is the smallest sufficient to stabilize convergence, minimizing the guiding policy bias.

\noindent\textbf{2) Three-vehicle platooning}. 
We evaluate MA-GPS in a nonlinear dynamic setting.
% and investigate the impact of guidance-policy bias. 
Initially positioned in separate highway lanes, Car 1 guides Cars 2 and 3 to a target lane. 
% Cost structures and model details appear in Appendix~\ref{sec:platooning-details}. 
% \textit{Implementation details:}\; See the full version~\cite{22}.
Figure~\ref{fig:three-vehicle} compares MA-GPS with state-of-the-art MARL methods. For fairness, training progress is reported against GPU time. Since exact Nash equilibria are difficult to compute in nonlinear games, policy quality is assessed through total rewards and trajectory simulations. MA-GPS consistently accelerates convergence and reduces variance during training, achieving variance reduction comparable to L2 regularization but with minimal bias, and learning high-performance policies significantly faster than baselines.

\noindent\textbf{3) Six-player strategic basketball formation.}
We demonstrate that MA-GPS scales to larger numbers of agents and converges faster to high-performing policies than IPPO without guidance, measured in training time. To illustrate, we simulate a simplified basketball formation where both offense and defense teams consist of three agents: a shooting guard, a center, and a forward. Defensive agents aim to stay between their assigned opponent and the basket. On offense, the center positions near the basket ($r=1.5$), the shooting guard holds the three-point line ($r=4$) to prepare for shots, and the forward sets a screen by moving toward the defender covering the shooting guard.

We note that some agents’ objectives include mild non-differentiabilities. For instance, a player’s objective may contain a term like $\max(p_y,0)$, where $p_y$ denotes that player’s y-position. This one-sided penalty introduces a kink at the baseline ($p_y=0$), discouraging the player from crossing it. Our experiments indicate that MA-GPS can accommodate such objectives beyond perfectly smooth settings, thereby broadening its applicability to a wider range of multi-agent tasks.

As shown in Figure~\ref{fig:basketball_all_plot}, MA-GPS noticeably increases the convergence rate and reduces variance compared to other four methods, {\color{black}particularly during the early stages of training, highlighting the stabilizing effect of the proposed model-based guidance}. As observed in the trajectory plot, as the number of agents increases, it becomes increasingly difficult for all agents to explore and find an effective policy together. Our guidance provides an effective prior for all agents, helping them converge to a reasonable equilibrium.

% \subsection{On the uniqueness of Nash equilibria and when MA-GPS outperforms baselines}

\begin{remark}[Comparing with finite-horizon iLQGames]
We measure wall-clock computation time $\tau$ for iLQGames with horizon length $T$ on two tasks: three-vehicle platooning (\(T=20, \tau=0.0485\pm0.0042\,\mathrm{s}\)) and six-player basketball formation (\(T=100, \tau=0.64\pm 0.15\,\mathrm{s}\)). Because iLQGames uses finite-horizon formulations, its computational time scales quickly with horizon length, player count, and objective complexity. In contrast, MA-GPS shifts computation to offline training and directly solves the infinite-horizon case, producing real-time executable neural network policies ($\tau = 0.0016\pm 0.0012\,\mathrm{s}$) for long-horizon strategic decision-making.
\end{remark}
% \begin{remark}[Computational time challenge of iLQGames]
% We measure wall-clock time for 20-step iLQGames on two tasks: three-vehicle platooning (\(\mu=0.0485\pm0.0042\,\mathrm{s}\)) and six-player basketball formation (\(\mu=46.81\pm4.18\,\mathrm{s}\)). As player count and objective complexity grow, iLQGames becomes impractical for real-time decisions. By contrast, MA-GPS shifts computation to training, enabling time-critical control. Moreover, as horizon length increases, the computational time further increases. However, our method directly considers infinite-horizon formulation, and therefore our policy allows the application to long-horizon strategic decision-making with real-time executable policies.
% \end{remark}

\section{Conclusion}
This paper introduced a multi-agent policy gradient method enhanced by model-based guidance, which stabilizes learning dynamics and accelerates convergence in non-cooperative dynamic games. The proposed approach mitigates the unstable limit-cycle behaviors that often arise in multi-agent policy gradient methods, thereby enabling the efficient learning of complex Nash equilibrium strategies.

Future work includes relaxing the assumption of an approximate model by learning dynamics models online and developing tighter theoretical bounds on the gap between guided and true Nash equilibria in LQ settings to better understand the impact of imperfect guidance on solution bias. Another direction is to unify MARL stabilization techniques, such as entropy-driven exploration, with model-based guidance to characterize how the balance between exploration and exploitation of game-theoretic structure affects the learning of Nash equilibrium policies.

% \balance

\bibliographystyle{IEEEtran} % use IEEEtran.bst style
\bibliography{reference}

\section{Appendix}
\label{sec:appendix}
% \begin{proposition}

% \begin{equation}
% \rho (T_\rho - \check{T}_\rho) = (I + S_\rho) \AKDelta + (\hat{S_\rho} - S_\rho) \AK + 
% (\hat{S_\rho} - S_\rho) \AKDelta
% \end{equation}

% \subsection{Proof of Proposition \ref{prop: L2 stabilization}}\label{appendix:proof of stable policy gradient dynamics}

\begin{lemma}\cite[Theorem 3]{mazumdar2020policy} \label{lemma:lq-nash-equation}
Consider a stabilizing policy $K = (K^1, \cdots, K^N)$.
% \begin{equation}
%     \Sigma_K = \mathbb{E}_{x_0} \Big[ \sum_{t=0}^{\infty} x_t x_t^{\top}\Big].
% \end{equation}
Then, $w(K) = 0$ if and only if $K$ is a Nash equilibrium. Moreover, the Nash equilibrium $K^{i*}$ satisfies $R^i K^{i*} = B^{i\top} P_{K^{*}}^{i} \bar{A}_{K^*}.$
% \vspace{-0.5em}
% \begin{equation}
% R^i K^{i*} = B^{i\top} P_{K^{*}}^{i} \bar{A}_{K^*}.
% \label{eq:lq-nash-condition}
% \end{equation}
\end{lemma}
% We then present a lemma, which is a variation of Lemma \ref{lemma:lq-nash-equation} required for the proof of Proposition~\ref{prop: L2 stabilization}.
\begin{lemma} 
\label{lemma:lq-nash-equation-hat}
Following the notations in Section~\ref{sec:lq-results}, we have 
$\hat{K}^{i*}$ satisfies $R^i \hat{K}^{i*} =\frac{1}{1 + \rho} (B^{i\top}  \hat{P}_{\hat{K}^*}^i \bar{A}_{\hat{K}^*} + \rho R^i \check{K}^i).$
% \begin{equation}
% (R^i + \rho I) \hat{K}^{i*} = B^{i\top}  \hat{P}_{\hat{K}^*}^i \bar{A}_{\hat{K}^*} + \rho \check{K}^i.
% \label{eq:lq-nash-condition-hat}
% \end{equation}
% \begin{equation}
% R^i \hat{K}^{i*} =\frac{1}{1 + \rho} \Big(B^{i\top}  \hat{P}_{\hat{K}^*}^i \bar{A}_{\hat{K}^*} + \rho R^i \check{K}^i\Big).
% \label{eq:lq-nash-condition-hat}
% \end{equation}
\end{lemma}
\begin{proof}
The proof is similar to the proof of Lemma \ref{lemma:lq-nash-equation}, with additional treatment of the regularization term. Under the feedback policy $\hat{K}$, we get $\sumcostregularized^i(\hat{K}^1,\dots,\hat{K}^N) \!= \!\expectationoverinit \Big[ x_0^\top\Big( (Q^i + {\hat{K}^i}{}^\top R^i \hat{K}^i ) + \rho {(\hat{K}^i - \check{K}^i)} {}^\top {R^i} (\hat{K}^i - \check{K}^i) + \bar{A}_{\hat{K}}^\top \hat{P}_{\hat{K}}^i \bar{A}_{\hat{K}} \Big) x_0 \Big].$ 
% \begin{equation}
%     \sumcostregularized^i(\hat{K}^1,\dots,\hat{K}^N) \!= \!\expectationoverinit \Big[ x_0^\top\Big( (Q^i + {\hat{K}^i}{}^\top R^i \hat{K}^i ) + \rho {(\hat{K}^i - \check{K}^i)} {}^\top {R^i} (\hat{K}^i - \check{K}^i) + \bar{A}_{\hat{K}}^\top \hat{P}_{\hat{K}}^i \bar{A}_{\hat{K}} \Big) x_0 \Big],
%     \label{eq:derive_condition_Khat1}
% \end{equation}
Taking the derivative of $\sumcostregularized^i(\hat{K}^1,\dots,\hat{K}^N)$ with respect to $\hat{K}^i$, we get $\nabla_{\hat{K}^i} \sumcostregularized^i(\hat{K}^1,\dots,\hat{K}^N) = 2 \big((1+\rho)R^i \hat{K}^i - \rho R^i \check{K}^i - {B^i}^{\top} \hat{P}_{\hat{K}}^i \bar{A}_{\hat{K}} \big) \Sigma_{\hat{K}}.$ Setting the gradient to zero completes the proof.
% \begin{equation}
%     \nabla_{\hat{K}^i} \sumcostregularized^i(\hat{K}^1,\dots,\hat{K}^N) = 2 \big((1+\rho)R^i \hat{K}^i - \rho R^i \check{K}^i - {B^i}^{\top} \hat{P}_{\hat{K}}^i \bar{A}_{\hat{K}} \big) \Sigma_{\hat{K}}.
% \label{eq:derive_condition_Khat2}    
% \end{equation}
\end{proof}

\begin{proof}[Proof of Proposition \ref{prop: L2 stabilization}]
\textbf{1)} The local Lipschitz continuity of $w(K)$ follows from Theorem 4.2 of \cite{mazumdar2020policy}. The existence of the set $\mathcal{K}_r$ is ensured because there exists a lower bound for a Lipschitz continuous function $w(K)$ over a compact set~\cite{rudin1987real}. 

\textbf{2)} Let $\tilde{r}$ be a radius such that $\{K:\|K-\check K\|_F\le \tilde{r}\}\subset\mathcal{K}$,
and let $\hat K^*$ be a regularized Nash equilibrium in this set. By Lemma \ref{lemma:lq-nash-equation-hat}, we
have
\begin{equation}\label{eq:latest proof}
\hat K^{i*}-\check K^i=\frac1\rho\,{R^i}^{-1}\Big(B^{i\top}\hat P^i_{\hat K^*}
\bar A_{\hat K^*}-R^i\hat K^{i*}\Big).
\end{equation}
As $\rho\to\infty$, the solution $\hat P^i_{\hat K^*}$ will not go to infinity, as
each player can always switch to $\check K^i$, for which the guidance cost
vanishes, and thereby achieve a finite and bounded cost: by the Nash property,
$\mathrm{tr}(\hat P^i_{\hat K^*}\Sigma_0)=\hat J^i_\rho(\hat K^*)
\le J^i(\check K^i,\hat K^{-i*})
\le \max_{\|K-\check K\|_F\le \tilde{r}}J^i(K)<\infty$, and $\Sigma_0\succ0$. Hence the
term in parentheses in \eqref{eq:latest proof} is bounded. Therefore,
$\lim_{\rho\to\infty}\|\hat K^*-\check K\|_2=0$.
 
By \eqref{eq:guided_lq_cost}, the regularized pseudo-gradient is $w+\rho\,g$, where $g$ is the
pseudo-gradient of the guidance term, so its Jacobian is $\nabla w+\rho\,\Gamma$
with $\Gamma:=\nabla g$. Since the guidance term of agent $i$ is quadratic in
$K^i-\check K^i$ with a weight $\Sigma_K$ that is smooth in $K$, we have
$g(\check K)=0$ and
$\Gamma(\check K)\succ0$, as given in item 2 in Proposition~\ref{prop: L2 stabilization}.
This also suggests that, as $\rho\to\infty$, we will eventually have
$\tfrac1\rho\big(\nabla w(\hat K^*)+\rho\,\Gamma(\hat K^*)\big)\to\Gamma(\check K)\succ0$,
so all eigenvalues of $\nabla w(\hat K^*)+\rho\,\Gamma(\hat K^*)$ have positive real
parts for sufficiently large $\rho$ (and $\hat K^*$ is locally stable for a
sufficiently small step size).
 
Finally, such an equilibrium exists for sufficiently large $\rho$ by the
implicit function theorem applied to $\tfrac1\rho w(K)+g(K)=0$ at $K=\check K$,
$1/\rho=0$, and its solution is a Nash equilibrium as in Lemma \ref{lemma:lq-nash-equation}, since the change of
variable $v=u^i+\tfrac{\rho}{1+\rho}\check K^i x$ turns each agent's problem in
\eqref{eq:guided_lq_cost} into an LQR problem with positive definite weights.

\textbf{3)} Finally, from Lemmas~\ref{lemma:lq-nash-equation} and \ref{lemma:lq-nash-equation-hat}, we get $R^i \left(\hat{K}^{i*} - \left(\frac{\rho}{1+\rho}\check{K}^i + \frac{1}{1+\rho}K^{i*} \right)\right)  = \frac{1}{1+\rho}{B^i}^\top (\hat{P}_{\hat{K}^*}^i \bar{A}_{\hat{K}^*} - P_{K^{*}}^{i} \bar{A}_{K^*}),$ 
% \begin{equation}
% R^i \left(\hat{K}^{i*} - \left(\frac{\rho}{1+\rho}\check{K}^i + \frac{1}{1+\rho}K^{i*} \right)\right)  = \frac{1}{1+\rho}{B^i}^\top (\hat{P}_{\hat{K}^*}^i \bar{A}_{\hat{K}^*} - P_{K^{*}}^{i} \bar{A}_{K^*}),
% \end{equation}
from which we get \eqref{eq:bias-bound}.
% Finally, by subtracting \eqref{eq:lq-nash-condition} from \eqref{eq:lq-nash-condition-hat}, we get
% \begin{equation}
% R^i (\hat{K}^{i*} - K^{i*}) + \rho \hat{K}^{i*} = {B^i}^\top (\hat{P}_{\hat{K}^*}^i \bar{A}_{\hat{K}^*} - P_{K^{*}}^{i} \bar{A}_{K^*}) + \rho \check{K}^i
% \end{equation}
% From this, we get
% \begin{equation}
% \hat{K}^{i*} - K^{i*} = (R^i + \rho I)^{-1}\big({B^i}^\top (\hat{P}_{\hat{K}^*}^i \bar{A}_{\hat{K}^*} - P_{K^{*}}^{i} \bar{A}_{K^*}) + \rho (\check{K}^i - K^{i*}) \big),
% \end{equation}
% to which we apply the triangular inequality of the matrix norm to get \eqref{eq:bias-bound}.
\end{proof}

\subsection{LQ Game Dynamics}
\label{sec:lq-AB}
\textbf{Model.} The system dynamics are represented by the state matrix $\statematrix$ and control matrices $\controlmatrix^1$ and $\controlmatrix^2$:
\[
\statematrix = 
\begin{bmatrix}
0.5880 & 0.0280 \\
0.5700 & 0.0560
\end{bmatrix},
\quad
\controlmatrix^1 = 
\begin{bmatrix}
1.0 \\
0.1
\end{bmatrix},
\quad
\controlmatrix^2 = 
\begin{bmatrix}
0 \\
1.0
\end{bmatrix}.
\]
\noindent \textbf{Costs.} The cost function is defined by the weighting matrices $Q^1$, $Q^2$, $R^1$, and $R^2$:
\begin{equation}
\begin{aligned}
&Q^1 = 
\begin{bmatrix}
0.0100 & 0 \\
0 & 1.0000
\end{bmatrix},
\quad
Q^2 = 
\begin{bmatrix}
1.0000 & 0 \\
0 & 0.0147
\end{bmatrix}, \\&  R^1 = 0.01,
\quad
R^2 = 0.01.
\end{aligned}
\end{equation}
The initial state is drawn from a distribution that assigns probability $0.5$ to $x_0 = [1,1]$ and probability $0.5$ to $x_0 = [1, 1.1]$. 
% \[
% R_1 = 0.0100,
% \quad
% R_2 = 0.0100.
% \]

\subsection{Vehicle platooning implementation details}
\label{sec:platooning-details}
\noindent \textbf{Model.} We model the cars as unicycles with the state ${x}^{i} = [p_{x}^{i}, p_{y}^{i}$, $ v^{i} , \theta^{i}]$, ${u}^{i} = [a^{i} , \omega^{i}]$ where  $p_{x}^{i}$, $p_{y}^{i}$, $ v^{i}$ and $\theta^{i}$ the x-coordinate, y-coordinate, velocity and heading angle of agent i $\in$ $\{1,2,3\}$. Car 1, the leader, aims to guide Car 2 and 3 to a target lane $p_x^{*}=0.5$.

\noindent \textbf{Costs.} The cost of each car is defined as follows
\begin{equation}
\begin{aligned}
c^1(x, u) = 
    & 5\big(\left( p_{x,t}^{2} - p_x^* \right)^2 + \left(p_{x,t}^{3} - p_x^* \right)^2\big) \\ & +
    \left( v_{t}^{1} - 1 \right)^2 +
    \left( \theta_{t}^{1} - \dfrac{\pi}{2} \right)^2 +
    {a_t^1}^2 +
    {\omega_t^1}^2,
\end{aligned}
\end{equation}
which is the cost of the leader vehicle, and
\begin{equation}
\begin{aligned}
c^i(x, u) = 
    & 5\left( p_{x,t}^{i} - p_{x,t}^{1} \right)^2 +
    \left(  v_{t}^{i} - 1 \right)^2 +
    \left(\theta_{t}^{i} - \dfrac{\pi}{2} \right)^2 \\& +
    {a_t^i}^2 +
    {\omega_t^i}^2
\end{aligned}
\end{equation}
for $i = 2, 3$, the following vehicle.

\noindent \textbf{Training.} The training for MA-GPS, IPPO and MA-PPO use the following parameters: 3 hidden layers and 512 hidden width (for both policy and value function networks), discount factor $\gamma = 0.99$, coefficient of entropy loss is 0.005, coefficient of squared-error loss $\left(V^{i,\pi}(x) - V^{\text{target}}\right)^2$ is 0.25, policy learning rate of $1 \times 10^{-4}$ and GAE $\lambda$ of 0.99. The reported results are from 5 random seeds for each method.

\subsection{ Six-player strategic basketball formation implementation details}
\label{sec:basketball-details}

\noindent \textbf{Model.} In this experiment, we model the player as double integrator with the state and the control input ${x}^{i} = [p_{x}^{i}, p_{y}^{i}$,$ v_{x}^{i} , v_{y}^{i}], {u}^{i} = [u_{x}^{i},u_y^i] $.

\noindent \textbf{Costs.} For i $\in \{1, 3\}$, representing the center and the shooting guard of the offender's team, the cost function $\stagecost^{i}(x, u)$ is defined as:
\begin{equation}
\begin{aligned}
    \stagecost^i(x, u) =& ({p_{x}^{i}}^2 + {p_{y}^{i}}^2 - r_i )^2 
    + ({v_{x}^{i}}^2 + {v_{y}^{i}}^2) \\ & + ({u_{x}^{i}}^2 + {u_{y}^{i}}^2) + 100 \max\{p_{y}^{i}, 0\},
    \label{eq:J_definition}
\end{aligned}
\end{equation}
where $r_i = 1.5, 4$ for $i =1, 3$, respectively. % is defined as:
% \[
% d_i = 
% \begin{cases} 
% 1.5^2, & \text{if } i = 1 \\
% 4^2, & \text{if } i = 3
% \end{cases}
% \]
For i $\in \{5\} $, representing the forward of the offender's team, the cost is defined as:
\begin{equation}
\stagecost^{i}(x, u) = \left[ (p_{x}^{i} -p_{x}^{i-1})^2 + (p_{y}^{i} - p_{y}^{i-1})^2 \right] 
+ (u_{x}^{i})^2 + (u_{y}^{i})^2,
\end{equation}
Aiming at screening the defense team player defending the guard. For i $\in \{2, 4, 6\}$, representing the defense team, the cost function $\stagecost^{i}(x, u)$ is defined as:
\begin{equation}
\begin{aligned}
\stagecost^{i}(x, u) = & \left[ (p_{x}^{i} - 0.75 p_{x}^{i-1})^2 + (p_{y}^{i} - 0.75 p_{y}^{i-1})^2 \right]  
\\& + (u_{x}^{i})^2 + (u_{y}^{i})^2.
\end{aligned}
\end{equation}

\noindent \textbf{Training.} The training for MA-GPS, IPPO and MA-PPO use the following parameters: 3 hidden layers and 512 hidden width (for both policy and value function networks), discount factor $\gamma = 0.99$, coefficient of entropy loss is 0.005, coefficient of squared-error loss $\left(V^{i,\pi}(x) - V^{\text{target}}\right)^2$ is 0.25, policy learning rate of $1 \times 10^{-4}$ and GAE $\lambda$ of 0.99. The reported results are from 5 random seeds for each method.

\end{document}